\documentclass[graybox]{svmult}

\usepackage{mathptmx}       
\usepackage{helvet}         
\usepackage{courier}        
\usepackage{type1cm}        
\usepackage{makeidx}         
\usepackage{graphicx}        
\usepackage{multicol}        
\usepackage[bottom]{footmisc}
\usepackage{amssymb}
\usepackage{amsfonts,amsmath}

\makeindex             

\begin{document}

\title*{Discrete versions of some  Dirac type equations and plane wave solutions}
\author{Volodymyr Sushch}
\institute{Volodymyr Sushch \at Koszalin University of Technology, Sniadeckich 2,
 75-453 Koszalin, Poland, \email{volodymyr.sushch@tu.koszalin.pl}}

%
%
\maketitle

\abstract*{A discrete version of the plane wave solution to some discrete Dirac type equations in the spacetime algebra is established.
The conditions under which a discrete analogue of the plane wave solution satisfies the discrete Hestenes equation are briefly discussed.}

\abstract {A discrete version of the plane wave solution to some discrete Dirac type equations in the spacetime algebra is established.
The conditions under which a discrete analogue of the plane wave solution satisfies the discrete Hestenes equation are briefly discussed.}

\section{Introduction}
\label{sec:1}
This work is a direct continuation of that described in my previous papers \cite{S3, S4}.
In \cite{S3}, a discrete analogue of the Dirac  equation for a free electron in the Hestenes form was constructed based on the discretization
scheme \cite{S2}. In \cite{S4}, a relationship between  the discrete   Dirac-K\"{a}hler  equation and discrete
analogues of some Dirac type equations in the spacetime algebra was discussed. In this paper, we establish a discrete version of plane wave solutions to discrete Dirac type equations.

We first briefly review some definitions and basic facts  on the Dirac-K\"{a}hler equation \cite{Kahler, Rabin} and the Dirac equation
 in the spacetime algebra \cite{H1, H2}.
Let $M={\mathbb R}^{1,3}$ be  Minkowski space with  metric signature  $(+,-,-,-)$.
Denote by $\Lambda^r(M)$ the vector space of smooth differential $r$-forms, $r=0,1,2,3,4$. We consider  $\Lambda^r(M)$ over $\mathbb{C}$.
Let $d:\Lambda^r(M)\rightarrow\Lambda^{r+1}(M)$ be the exterior differential and let $\delta:\Lambda^r(M)\rightarrow\Lambda^{r-1}(M)$ be the formal adjoint of $d$  with respect to  the natural inner product in $\Lambda^r(M)$. We have $\delta=\ast d\ast$, where  $\ast$ is the Hodge star operator  $\ast:\Lambda^r(M)\rightarrow\Lambda^{4-r}(M)$ with respect to the Lorentz metric.
 Denote by $\Lambda(M)$ the set of all differential forms on $M$. We have
\begin{equation*}
\Lambda(M)=\Lambda^{ev}(M)\oplus\Lambda^{od}(M),
\end{equation*}
where $\Lambda^{ev}(M)=\Lambda^0(M)\oplus\Lambda^2(M)\oplus\Lambda^4(M)$  and  $\Lambda^{od}(M)=\Lambda^1(M)\oplus\Lambda^3(M)$.

Let $\Omega\in\Lambda(M)$
be an inhomogeneous differential form, i.e.
$\Omega=\sum_{r=0}^4\overset{r}{\omega},$
where $\overset{r}{\omega}\in\Lambda^r(M)$.
 The Dirac-K\"{a}hler equation is given by
\begin{equation}\label{eq:01}
i(d+\delta)\Omega=m\Omega,
\end{equation}
where $i$ is the usual complex unit  and  $m$  is a mass parameter.
It is easy to show that Eq.~(\ref{eq:01}) is equivalent to the set of equations
\begin{eqnarray*}\label{}
i\delta\overset{1}{\omega}=m\overset{0}{\omega},\\
i(d\overset{0}{\omega}+\delta\overset{2}{\omega})=m\overset{1}{\omega},\\
i(d\overset{1}{\omega}+\delta\overset{3}{\omega})=m\overset{2}{\omega},\\
i(d\overset{2}{\omega}+\delta\overset{4}{\omega})=m\overset{3}{\omega},\\
id\overset{3}{\omega}=m\overset{4}{\omega}.
\end{eqnarray*}
The operator $d+\delta$ is an analogue of the gradient operator $\nabla=\sum_{\mu=0}^3\gamma_\mu\partial^\mu$ in Minkowski spacetime,   where $\gamma_\mu$ is the Dirac gamma matrix and $\partial^\mu$ is a partial derivative. Think of $\{\gamma_0, \gamma_1, \gamma_2, \gamma_3\}$ as a vector basis in spacetime. Then the  gamma matrices $\gamma_\mu$ can be considered as generators of the Clifford algebra  $\emph{C}\ell(1,3)$  \cite{B, B1}. This algebra Hestenes \cite{H2} calls  the spacetime algebra. Denote by $\emph{C}\ell_{\mathbb{R}}(1,3)$ $(\emph{C}\ell_{\mathbb{C}}(1,3))$ the real (complex) Clifford algebra.  It is known that an inhomogeneous form $\Omega$ can be represented as element of $\emph{C}\ell_{\mathbb{C}}(1,3)$. Then the Dirac-K\"{a}hler equation can be written as the algebraic equation
\begin{equation}\label{eq:02}
  i\nabla\Omega=m\Omega, \quad \Omega\in\emph{C}\ell_{\mathbb{C}}(1,3).
 \end{equation}
  Eq.~(\ref{eq:02}) is equivalent to the four Dirac equations (traditional column-spinor equations) for a free electron.  Let $\emph{C}\ell^{ev}(1,3)$ be the even subalgebra of the  algebra $\emph{C}\ell(1,3)$. The equation
 \begin{equation}\label{eq:03}
-\nabla\Omega^{ev}\gamma_1\gamma_2=m\Omega^{ev}\gamma_0, \quad \Omega^{ev}\in\emph{C}\ell_{\mathbb{R}}^{ev}(1,3),
 \end{equation}
is called the Hestenes form of the Dirac equation \cite{H1, H2}.
Consider also  the equation
\begin{equation}\label{eq:04}
 i\nabla\Omega^{ev}=m\Omega^{ev}\gamma_0 , \quad \Omega^{ev}\in\emph{C}\ell_{\mathbb{C}}^{ev}(1,3).
 \end{equation}
 In  \cite{J},  this equation  is called  the "generalized bivector Dirac equation". Following Baylis \cite{B1} we call Eq.~(\ref{eq:04}) the Joyce equation. This equation admits the plane wave solution of the form
 \begin{equation}\label{eq:05}
 \Psi=A\exp\Big(i\sum_{\mu=0}^3 p^\mu x_\mu\Big),
 \end{equation}
 where $A\in\emph{C}\ell_{\mathbb{C}}^{ev}(1,3)$ is a constant element and  $\{p^0, p^1, p^2, p^3\}$ is a four-momentum.
 Suppose that for exterior forms (elements of $\Lambda(M)$) the Clifford multiplication is defined. It should be noted that the graded algebra
$\Lambda(M)$ endowed with the Clifford multiplication is an example of the  Clifford
algebra. In this case the basis covectors $e^\mu=dx^\mu$, $\mu=0,1,2,3$, of spacetime are considered as generators of the Clifford algebra.  Let $\Lambda_{\mathbb{R}}(M)$ denote the set of real-valued differential forms. Then Eqs.~(\ref{eq:03}) and (\ref{eq:04}) can be rewritten in terms of inhomogeneous forms as
\begin{equation}\label{eq:06}
-(d+\delta)\Omega^{ev} e^1e^2=m\Omega^{ev} e^0, \quad \Omega^{ev}\in\Lambda_{\mathbb{R}}^{ev}(M),
 \end{equation}
 and
\begin{equation}\label{eq:07}
 i(d+\delta)\Omega^{ev}=m\Omega^{ev} e^0, \quad \Omega^{ev}\in\Lambda^{ev}(M).
 \end{equation}

The aim of the present paper is to construct a discrete version of the plane wave solution (\ref{eq:05}). In much the same way as in the continuum case \cite{B1,J} we show that the discrete Joyce equation admits eight linearly independent plane wave solutions in the discrete formulation.  We briefly discuss the conditions under which  the obtained  plane wave solutions satisfy the discrete Hestenes equation.

\section{Discrete Dirac-K\"{a}hler equation}
\label{sec:2}
In this section, we start off with a  discretization scheme.
The scheme is based on the language of differential forms and  is described in  \cite{S2}.This approach was originated by Dezin in \cite{Dezin}.
Due to space limitations, we skip  the relevant material from  \cite{S2}. For the convenience of the reader, we fix only some notation and recall some facts   concerning  discrete analogues of the differential operators  $d$ and $\delta$.   All details can be found in   \cite{S1, S2}.

Let $K(4)=K\otimes K\otimes K\otimes K$
be a cochain complex with  complex  coefficients,
where  $K$ is  the 1-dimensional complex generated by 0- and 1-dimensional basis elements   $x^{\kappa}$  and $e^{\kappa}$,  $\kappa\in\mathbb{Z}$,  respectively.
Then an arbitrary $r$-dimensional basis element of $K(4)$ can be written as
$s^k_{(r)}=s^{k_0}\otimes s^{k_1}\otimes s^{k_2}\otimes s^{k_3}$, where
$s^{k_\mu}$ is either $x^{k_\mu}$ or $e^{k_\mu}$,  $k=(k_0,k_1,k_2,k_3)$ and \ $k_\mu\in\mathbb{Z}$.
Let
\begin{equation*}
x^k=x^{k_0}\otimes x^{k_1}\otimes x^{k_2}\otimes x^{k_3}, \qquad e^k=e^{k_0}\otimes e^{k_1}\otimes e^{k_2}\otimes e^{k_3}
\end{equation*}
denote the 0- and 4-dimensional basis elements of $K(4)$.
The dimension $r$ of a basis element $s^k_{(r)}$ is given
by the number of factors $e^{k_\mu}$ that appear in it. For example, the 1-dimensional basis element
$e^k_\mu\in K(4)$ can be written as
\begin{eqnarray*}
e^k_0=e^{k_0}\otimes x^{k_1}\otimes x^{k_2}\otimes x^{k_3},  \qquad
e^k_1=x^{k_0}\otimes e^{k_1}\otimes x^{k_2}\otimes x^{k_3}, \\
e^k_2=x^{k_0}\otimes x^{k_1}\otimes e^{k_2}\otimes x^{k_3},  \qquad
e^k_3=x^{k_0}\otimes x^{k_1}\otimes x^{k_2}\otimes e^{k_3},
\end{eqnarray*}
where  the subscript $\mu=0,1,2,3$ indicates  a place of $e^{k_\mu}$ in $e^k_\mu$.
Similarly,   $e_{\mu\nu}^k$, $\mu<\nu$, and  $e_{\iota\mu\nu}^k$, $\iota<\mu<\nu$, denote the   2- and 3-dimensional basic elements of $K(4)$.
The complex $K(4)$ is a discrete analogue of $\Lambda(M)$ and cochains play the role of differential
forms. Let us call  them forms or discrete forms to emphasize their relationship with differential
forms. Denote by  $K^r(4)$ the set of all $r$-forms. Then we have
\begin{equation*}
K(4)=K^{ev}(4)\oplus K^{od}(4),
\end{equation*}
 where $K^{ev}(4)=K^0(4)\oplus K^2(4)\oplus K^4(4)$ and $K^{od}(4)=K^1(4)\oplus K^3(4)$.
 Any $r$-form $\overset{r}{\omega}\in K^r(4)$ can be expressed as
\begin{eqnarray}\label{eq:08}
\overset{0}{\omega}=\sum_k\overset{0}{\omega}_kx^k, \qquad \overset{2}{\omega}=\sum_k\sum_{\mu<\nu} \omega_k^{\mu\nu}e_{\mu\nu}^k,  \qquad  \overset{4}{\omega}=\sum_k\overset{4}{\omega}_ke^k, \\ \label{eq:09}
\overset{1}{\omega}=\sum_k\sum_{\mu=0}^3\omega_k^\mu e_\mu^k, \qquad
\overset{3}{\omega}=\sum_k\sum_{\iota<\mu<\nu} \omega_k^{\iota\mu\nu}e_{\iota\mu\nu}^k,
\end{eqnarray}
where  $\overset{0}{\omega}_k, \ \omega_k^{\mu\nu}, \ \overset{4}{\omega}_k, \  \omega_k^\mu$ and $\omega_k^{\iota\mu\nu}$ are complex numbers.
A discrete inhomogeneous form $\Omega\in K(4)$  is defined to be
\begin{equation}\label{eq:10}
\Omega=\sum_{r=0}^4\overset{r}{\omega}.
\end{equation}
Let $d^c: K^r(4)\rightarrow K^{r+1}(4)$ be a discrete analogue of the exterior derivative $d$ and let $\delta ^c: K^r(4)\rightarrow K^{r-1}(4)$ be a discrete analogue of the codifferential $\delta$. For more precise  definitions of these operators we refer the reader to  \cite{S2}. In this paper  we give only the difference
expressions for $d^c$ and  $\delta ^c$.
Let the difference operator $\Delta_\mu$ be defined by
\begin{equation}\label{eq:11}
\Delta_\mu\omega_k^{(r)}=\omega_{\tau_\mu k}^{(r)}-\omega_k^{(r)},
\end{equation}
where  $\omega_k^{(r)}\in\mathbb{C}$ is a component of $\overset{r}{\omega}\in K^r(4)$ and
$\tau_\mu$ is   the shift operator  which acts  as
$\tau_\mu k=(k_0,...k_\mu+1,...k_3), \   \mu=0,1,2,3.$
For forms (\ref{eq:08}),  (\ref{eq:09})  we have
\begin{eqnarray}\label{eq:12}
d^c\overset{0}{\omega}=\sum_k\sum_{\mu=0}^3(\Delta_\mu\overset{0}{\omega}_k)e_\mu^k,  \qquad d^c\overset{1}{\omega}=\sum_k\sum_{\mu<\nu}(\Delta_\mu\omega_k^\nu-\Delta_\nu\omega_k^\mu)e_{\mu\nu}^k,
\end{eqnarray}
\begin{eqnarray}\label{eq:13}
d^c\overset{2}{\omega}=\sum_k\big[(\Delta_0\omega_k^{12}-\Delta_1\omega_k^{02}+\Delta_2\omega_k^{01})e_{012}^k
+(\Delta_0\omega_k^{13}-\Delta_1\omega_k^{03}+\Delta_3\omega_k^{01})e_{013}^k \nonumber \\
+(\Delta_0\omega_k^{23}-\Delta_2\omega_k^{03}+\Delta_3\omega_k^{02})e_{023}^k
+(\Delta_1\omega_k^{23}-\Delta_2\omega_k^{13}+\Delta_3\omega_k^{12})e_{123}^k\big],
\end{eqnarray}
\begin{equation}\label{eq:14}
d^c\overset{3}{\omega}=\sum_k(\Delta_0\omega_k^{123}-\Delta_1\omega_k^{023}+\Delta_2\omega_k^{013}-\Delta_3\omega_k^{012})e^k, \qquad d^c\overset{4}{\omega}=0,
\end{equation}
\begin{equation}\label{eq:15}
\delta^c\overset{0}{\omega}=0, \qquad \delta^c\overset{1}{\omega}=\sum_k(\Delta_0\omega_k^{0}-\Delta_1\omega_k^{1}-\Delta_2\omega_k^{2}-\Delta_3\omega_k^{3})x^k,
\end{equation}
\begin{eqnarray}\label{eq:16} \nonumber
\delta^c\overset{2}{\omega}=\sum_k\big[(\Delta_1\omega_k^{01}+\Delta_2\omega_k^{02}+\Delta_3\omega_k^{03})e_{0}^k
+(\Delta_0\omega_k^{01}+\Delta_2\omega_k^{12}+\Delta_3\omega_k^{13})e_{1}^k\\
+(\Delta_0\omega_k^{02}-\Delta_1\omega_k^{12}+\Delta_3\omega_k^{23})e_{2}^k
+(\Delta_0\omega_k^{03}-\Delta_1\omega_k^{13}-\Delta_2\omega_k^{23})e_{3}^k\big],
\end{eqnarray}
\begin{eqnarray}\label{eq:17} \nonumber
\delta^c\overset{3}{\omega}=\sum_k\big[(-\Delta_2\omega_k^{012}-\Delta_3\omega_k^{013})e_{01}^k+
(\Delta_1\omega_k^{012}-\Delta_3\omega_k^{023})e_{02}^k\\ \nonumber
+(\Delta_1\omega_k^{013}+\Delta_2\omega_k^{023})e_{03}^k
+(\Delta_0\omega_k^{012}-\Delta_3\omega_k^{123})e_{12}^k\\
+(\Delta_0\omega_k^{013}+\Delta_2\omega_k^{123})e_{13}^k
+(\Delta_0\omega_k^{023}-\Delta_1\omega_k^{123})e_{23}^k\big],
\end{eqnarray}
\begin{eqnarray}\label{eq:18}
\delta^c\overset{4}{\omega}=\sum_k\big[(\Delta_3\overset{4}{\omega}_k)e_{012}^k-(\Delta_2\overset{4}{\omega}_k)e_{013}^k
+(\Delta_1\overset{4}{\omega}_k)e_{023}^k+(\Delta_0\overset{4}{\omega}_k)e_{123}^k\big].
\end{eqnarray}
Let $\Omega\in K(4)$ be given by (\ref{eq:10}). A discrete analogue of the Dirac-K\"{a}hler equation (\ref{eq:01}) can be defined as
 \begin{equation}\label{eq:19}
i(d^c+\delta^c)\Omega=m\Omega.
\end{equation}
We can write this equation more explicitly by separating its homogeneous components as
\begin{eqnarray}\label{eq:20} \nonumber
i\delta^c\overset{1}{\omega}=m\overset{0}{\omega}, \quad i(d^c\overset{1}{\omega}+\delta^c\overset{3}{\omega})=m\overset{2}{\omega}, \quad
id^c\overset{3}{\omega}=m\overset{4}{\omega},\\
i(d^c\overset{0}{\omega}+\delta^c\overset{2}{\omega})=m\overset{1}{\omega}, \qquad
i(d^c\overset{2}{\omega}+\delta^c\overset{4}{\omega})=m\overset{3}{\omega}.
\end{eqnarray}
Substituting (\ref{eq:12})--(\ref{eq:18})  into (\ref{eq:20}) one  obtains the set of 16 difference equations   \cite{S2}.
\section{Discrete Hestenes and Joyce equations}
\label{sec:3}
As in \cite{S3},  we define  the Clifford multiplication of the basis elements $x^k$ and $e^k_\mu$,  \  $\mu=0,1,2,3$, by the following rules:

\medskip
(a) $x^kx^k=x^k, \quad x^ke^k_\mu=e^k_\mu x^k=e^k_\mu$;

(b) $e^k_\mu e^k_\nu+e^k_\nu e^k_\mu=2g_{\mu\nu}x^k$;

(c) $e^k_{\mu_1}\cdots e^k_{\mu_s}=e^k_{\mu_1\cdots \mu_s}$ \ for \ $0\leq \mu_1<\cdots <\mu_s\leq 3$.
\medskip

Here  $g_{\mu\nu}=diag(1,-1,-1,-1)$ is the metric tensor.
Note that the multiplication is defined for the basis elements of $K(4)$ with the same multi-index $k=(k_0,k_1,k_2,k_3)$ supposing the product to be zero in all other cases.
The operation is linearly extended to arbitrary discrete forms.

Consider the following unit forms
\begin{equation}\label{eq:21}
x=\sum_kx^k, \qquad e=\sum_ke^k, \qquad e_\mu=\sum_ke_\mu^k, \qquad e_{\mu\nu}=\sum_ke_{\mu\nu}^k,
\end{equation}
where $\mu, \nu=0,1,2,3$.
The unit 0-form $x$ plays  a role of the unit element in $K(4)$, i.e. for any $r$-form  $\overset{r}{\omega}$ we have $x\overset{r}{\omega}=\overset{r}{\omega}x=\overset{r}{\omega}$.
\begin{proposition}
The following holds:
\begin{equation}\label{eq:22}
e_\mu e_\nu+e_\nu e_\mu=2g_{\mu\nu}x, \qquad \mu,\nu=0,1,2,3.
\end{equation}
\end{proposition}
\begin{proof}
By the rule (b) it is obvious.
\end{proof}
\begin{proposition}
For any inhomogeneous form $\Omega\in K(4)$ we have
\begin{equation}\label{eq:23}
(d^c+\delta^c)\Omega=\sum_{\mu=0}^3e_\mu\Delta_\mu\Omega,
\end{equation}
where
 $\Delta_\mu$ is the difference operator which acts on each component of $\Omega$ by the rule~(\ref{eq:11}).
\end{proposition}
\begin{proof}
See Proposition~1 in  \cite{S4}.
\end{proof}
Thus the discrete Dirac-K\"{a}hler equation can be rewritten in the form
\begin{equation*}
i\sum_{\mu=0}^3e_\mu\Delta_\mu\Omega=m\Omega.
\end{equation*}
Let $\Omega^{ev}\in K^{ev}(4)$ be a real-valued even inhomogeneous form, i.e. $\Omega^{ev}=\overset{0}{\omega}+\overset{2}{\omega}+\overset{4}{\omega}$.  A discrete analogue of the Hestenes equation (\ref{eq:06}) is defined by
\begin{equation}\label{eq:24}
-(d^c+\delta^c)\Omega^{ev} e_1e_2=m\Omega^{ev} e_0,
\end{equation}
or equivalently,
\begin{equation*}
-\sum_{\mu=0}^3e_\mu\Delta_\mu\Omega^{ev}e_1e_2=m\Omega^{ev} e_0,
\end{equation*}
where $e_1, e_2$ and $e_0$ are given by  (\ref{eq:21}).
A discrete analogue of the Joyce equation (\ref{eq:07}) is given by
\begin{equation}\label{eq:25}
i(d^c+\delta^c)\Omega^{ev}=m\Omega^{ev} e_0,
\end{equation}
where $\Omega^{ev}\in K^{ev}(4)$ is a complex-valued even inhomogeneous form.
Clearly, Eq.~(\ref{eq:25}) can be rewritten in the form
\begin{equation*}
i\sum_{\mu=0}^3e_\mu\Delta_\mu\Omega^{ev}=m\Omega^{ev} e_0.
\end{equation*}
Applying (\ref{eq:12})--(\ref{eq:18}) Eqs.~(\ref{eq:24}) and (\ref{eq:25})   can be expressed also in terms of difference equations (see \cite{S4}).

Consider  the following constant forms
\begin{equation}\label{eq:26}
P_{\pm 0}=\frac{1}{2}(x\pm e_0), \qquad  P_{\pm 12}=\frac{1}{2}(x\pm ie_1e_2).
\end{equation}
   Since
 \begin{equation*}
(P_{\pm 0})^2=P_{\pm 0}P_{\pm 0}=P_{\pm 0}, \qquad  (P_{\pm 12})^2=P_{\pm 12}P_{\pm 12}=P_{\pm 12},
\end{equation*}
  it follows that  $P_{\pm 0}$ and $P_{\pm 12}$ are projectors.
The projectors $P_{\pm 0}$ and $P_{\pm 12}$ have the following properties:
\begin{equation}\label{eq:27}
P_{\pm 0}P_{\pm 12}=P_{\pm 12}P_{\pm 0}, \quad e_0P_{\pm 0}=P_{\pm 0}e_0,  \quad e_1e_2P_{\pm 12}=P_{\pm 12}e_1e_2,
\end{equation}
\begin{equation}\label{eq:28}
P_{\pm 0}=\pm P_{\pm 0}e_0,  \qquad  P_{\pm 12}=\pm iP_{\pm 12}e_1e_2.
\end{equation}
See \cite{S3} for more details.

Recall that the  Hestenes equation is defined on real-valued even  forms.
 First suppose that the discrete Hestenes equation (\ref{eq:24}) acts in $K(4)$, i.e. acts in the same space as  the discrete Dirac-K\"{a}hler equation.
\begin{proposition}
Let $\Omega^{ev}\in K^{ev}(4)$ be a solution of the discrete Joyce equation,  then
\begin{equation*}\label{3.12}
\Omega^{ev}=\Omega^{ev} P_{+0}+\Omega^{ev} P_{-0},
\end{equation*}
where
$\Omega^{ev} P_{+0}$  satisfies the discrete Dirac-K\"{a}hler equation while
$\Omega^{ev} P_{-0}$  satisfies the same equation but the sign of the right-hand side changed to its opposite.
\end{proposition}
\begin{proposition}
Let $\Omega^{ev}\in K^{ev}(4)$ be a solution of the discrete Joyce equation,  then
\begin{equation*}\label{3.13}
\Omega^{ev}=\Omega^{ev} P_{+12}+\Omega^{ev} P_{-12},
\end{equation*}
where
$\Omega^{ev} P_{+12}$  satisfies the discrete Hestenes equation  while
$\Omega^{ev} P_{-12}$  satisfies the same equation but the sign of the right-hand side changed to its opposite.
\end{proposition}
By (\ref{eq:27}) and (\ref{eq:28}), the proof is straightforward.

In the case of the real-defined discrete Hestenes equation, we have the following results. It is proven in  \cite[Proposition~5]{S4} that by a solution of the discrete Dirac-K\"{a}hler equation four independent solutions of the discrete Hestenes  equation (\ref{eq:24}) are constructed.
Every solution  of the discrete Joyce equation can be represented in the form in which  each term of the real and imaginary parts  is a real even solution   of the discrete Hestenes equation with the correct or reversed sign on the right-hand side \cite[Proposition~6]{S3}. These are discrete versions of well-known results for corresponding continuum equations.

\section{Plane wave solutions}
\label{sec:}
Let us consider the following 0-form
\begin{equation}\label{eq:29}
 \psi=\sum_k\psi_kx^k,
 \end{equation}
 where
 \begin{equation}\label{eq:30}
  \psi_k=(ip_0+1)^{k_0}(ip_1+1)^{k_1}(ip_2+1)^{k_2}(ip_3+1)^{k_3}, \qquad p_\mu\in\mathbb{R}.
 \end{equation}
We wish to find a solution of the discrete Joyce equation of the form

\begin{equation}\label{eq:31}
 \Omega=A\psi,
 \end{equation}
 where
   $A\in K^{ev}(4)$ is an inhomogeneous constant form. A constant form means that its components do not depend on $k$.  More explicitly, let
     \begin{equation*}
  A=\overset{0}{\alpha}+\overset{2}{\alpha}+\overset{4}{\alpha},
 \end{equation*}
where
\begin{equation*}
  \overset{0}{\alpha}=\alpha^0x, \qquad \overset{2}{\alpha}=\sum_{\mu<\nu}\alpha^{\mu\nu}e_{\mu\nu}, \qquad \overset{4}{\alpha}=\alpha^4e,
   \end{equation*}
and $\alpha^0, \alpha^{\mu\nu}, \alpha^4\in\mathbb{C}$. Recall that  $x$, $e_{\mu\nu}$ and $e$ are the unit forms given by (\ref{eq:21}).
The form  (\ref{eq:31}), where the components of $\psi$ are given by (\ref{eq:30}),  is a discrete version  of the plane wave solution (\ref{eq:05}).

It is easy to check that
\begin{equation*}
  \Delta_\mu\psi_k=ip_\mu\psi_k, \quad \mu=0,1,2,3.
   \end{equation*}
   Consequently, the factor $ip_\mu$ is an eigenvalue of the difference operator $\Delta_\mu$.
   This clearly forces
\begin{equation*}
   d^c\psi=\sum_k\sum_{\mu=0}^3(ip_\mu\psi_k)e_\mu^k.
 \end{equation*}
According to (\ref{eq:23}) we have
 \begin{eqnarray*}
(d^c+\delta^c)\Omega=(d^c+\delta^c)A\psi=\sum_{\mu=0}^3e_\mu\Delta_\mu (A\psi)=\sum_{\mu=0}^3e_\mu A(\Delta_\mu\psi)\\
=\sum_{\mu=0}^3e_\mu A\Big(\sum_k(\Delta_\mu\psi_k) x^k\Big)=i\sum_{\mu=0}^3e_\mu p_\mu A\Big(\sum_k\psi_k x^k\Big)=i\Big(\sum_{\mu=0}^3e_\mu p_\mu\Big) A\psi.
\end{eqnarray*}
Substituting  into equation (\ref{eq:25})  we obtain
 \begin{equation*}
   -\Big(\sum_{\mu=0}^3e_\mu p_\mu\Big) A\psi=mA\psi e_0.
 \end{equation*}
Since  $\psi e_0=e_0\psi$, this equation  reduces to
 \begin{equation}\label{eq:32}
   -\Big(\sum_{\mu=0}^3e_\mu p_\mu\Big) A=mAe_0.
 \end{equation}
From (\ref{eq:22}) it follows that $e_0e_0=x$. Then Eq.~(\ref{eq:32}) can be written as
 \begin{equation}\label{eq:33}
   -\Big(p_0x+\sum_{\mu=1}^3p_\mu e_0 e_\mu\Big)A=me_0A e_0.
 \end{equation}
 By (\ref{eq:22}) and by  trivial computation, we have
 \begin{equation*}
   \Big(p_0x-\sum_{\mu=1}^3p_\mu e_0 e_\mu\Big)\Big(p_0x+\sum_{\mu=1}^3p_\mu e_0 e_\mu\Big)=\Big(p_0^2-\sum_{\mu=1}^3p_\mu^2\Big)x.
 \end{equation*}
 Therefore, multiplying both sides of (\ref{eq:33}) by the same factor $-\big(p_0x-\sum_{\mu=1}^3p_\mu e_0 e_\mu\big)$ gives
 \begin{equation*}
  \Big(p_0^2-\sum_{\mu=1}^3p_\mu^2\Big)xA=-m\Big(p_0x-\sum_{\mu=1}^3p_\mu e_0 e_\mu\Big)e_0A e_0.
 \end{equation*}
 This implies
  \begin{equation*}
  \Big(p_0^2-\sum_{\mu=1}^3p_\mu^2\Big)A=-m\Big(\sum_{\mu=0}^3p_\mu e_\mu\Big)A e_0.
 \end{equation*}
 Applying (\ref{eq:32}) again  we obtain
 \begin{equation*}
  \Big(p_0^2-\sum_{\mu=1}^3p_\mu^2\Big)A=m^2A e_0e_0,
 \end{equation*}
 or equivalently,
 \begin{equation*}
  \Big(p_0^2-\sum_{\mu=1}^3p_\mu^2\Big)A=m^2A.
 \end{equation*}
  Thus we have the following assertion.
  \begin{proposition}
  The form (\ref{eq:31}) is a non-trivial solution of Eq.~(\ref{eq:25}) if and only if
  \begin{equation*}
  p_0^2=\sum_{\mu=1}^3p_\mu^2+m^2,
 \end{equation*}
 or equivalently,
  \begin{equation}\label{eq:34}
   p_0=\pm\sqrt{m^2+p_1^2+p_2^2+p_3^2}.
 \end{equation}
  \end{proposition}
  The condition (\ref{eq:34}) is the same as in the case of continuum counterpart \cite{B1}.

Let us represent the even form  $A$  as
\begin{equation*}
  A=A_{+}+A_{-},
 \end{equation*}
 where
 \begin{equation}\label{eq:35}
   A_{+}=\alpha^{0}x+\alpha^{12}e_{12}+\alpha^{13}e_{13}+\alpha^{23}e_{23},
  \end{equation}
 \begin{equation}\label{eq:36}
   A_{-}=\alpha^{01}e_{01}+\alpha^{02}e_{02}+\alpha^{03}e_{03}+\alpha^4e.
  \end{equation}
  It is easy to check that $A_{+}$ commutes with $e_0$  and  $A_{-}$  anticommutes with it, i.e.
  \begin{equation}\label{eq:37}
  e_0A_{\pm}=\pm A_{\pm}e_0.
 \end{equation}
 \begin{lemma}
  The form $e_{0\mu}A_{-}$ commutes with $e_0$ and has the view   (\ref{eq:35}),  while
  $e_{0\mu}A_{+}$ anticommutes with $e_0$ and has the view   (\ref{eq:36}) for any $\mu=1,2,3$.
  \end{lemma}
  \begin{proof} For $\mu=1$ we have
  \begin{eqnarray*}
   e_{01}A_{-}=e_{01}(\alpha^{01}e_{01}+\alpha^{02}e_{02}+\alpha^{03}e_{03}+\alpha^4e)\\
   =\alpha^{01}x-\alpha^{02}e_{12}-\alpha^{03}e_{13}+\alpha^4e_{23}.
  \end{eqnarray*}
  The same proof remains valid for all other cases.
  \end{proof}
  \begin{theorem}
 The form  $\Omega=A\psi$ is  a  non-trivial solution of the discrete Joyce equation if and only if the condition
\begin{equation}\label{eq:38}
 A_{-}=\frac{p_1e_0e_1+p_2e_0e_2+p_3e_0e_3}{m-p_0}A_{+}
 \end{equation}
 holds,  or equivalently,
 \begin{equation}\label{eq:39}
 A_{+}=-\frac{p_1e_0e_1+p_2e_0e_2+p_3e_0e_3}{m+p_0}A_{-}.
 \end{equation}
\end{theorem}
\begin{proof}
Let $\Omega=A\psi$ satisfy (\ref{eq:25}). Then we have
\begin{equation*}
   -\Big(\sum_{\mu=0}^3e_\mu p_\mu\Big)(A_{+}+A_{-})=m(A_{+}+A_{-})e_0,
 \end{equation*}
 or
 \begin{equation*}
   -\Big(\sum_{\mu=1}^3e_0e_\mu p_\mu\Big)(A_{+}+A_{-})=p_0(A_{+}+A_{-})+me_0(A_{+}+A_{-})e_0.
 \end{equation*}
 Applying (\ref{eq:37}) we can rewrite the above relationship as
 \begin{equation*}
   -\Big(\sum_{\mu=1}^3e_0e_\mu p_\mu\Big)(A_{+}+A_{-})=(p_0+m)A_{+}+(p_0-m)A_{-}.
 \end{equation*}
 By Lemma~1, collecting like terms gives
 \begin{equation*}
   -(e_0e_1 p_1+e_0e_2 p_2+e_0e_3 p_3)A_{+}=(p_0-m)A_{-},
 \end{equation*}
 \begin{equation*}
   -(e_0e_1 p_1+e_0e_2 p_2+e_0e_3 p_3)A_{-}=(p_0+m)A_{+}.
 \end{equation*}
 Conversely, substituting (\ref{eq:38}) into (\ref{eq:39}) yields the condition (\ref{eq:34}). It follows that $A\psi$ is  a  non-trivial solution of (\ref{eq:25}).
 \end{proof}
  Note that the condition (\ref{eq:38}) can be rewritten as the following system of equations
 \begin{eqnarray*}
(m-p_0)\alpha^{01}-p_1\alpha^{0}-p_2\alpha^{12}-p_3\alpha^{13}=0, \\
  (m-p_0)\alpha^{02}-p_2\alpha^{0}+p_1\alpha^{12}-p_3\alpha^{23}=0, \\
    (m-p_0)\alpha^{03}-p_3\alpha^{0}+p_1\alpha^{13}+p_2\alpha^{23}=0, \\
 (m-p_0)\alpha^4-p_1\alpha^{23}+p_1\alpha^{13}-p_3\alpha^{12}=0.
 \end{eqnarray*}
 Similarly, the condition (\ref{eq:39}) gives the following equivalent system of equations
 \begin{eqnarray*}
(m+p_0)\alpha^{0}+p_1\alpha^{01}+p_2\alpha^{02}+p_3\alpha^{03}=0, \\
  (m+p_0)\alpha^{12}-p_1\alpha^{02}+p_2\alpha^{01}+p_3\alpha^{4}=0, \\
    (m+p_0)\alpha^{13}-p_1\alpha^{03}-p_2\alpha^{4}+p_3\alpha^{01}=0, \\
 (m+p_0)\alpha^{23}+p_1\alpha^{4}-p_2\alpha^{03}+p_3\alpha^{02}=0.
 \end{eqnarray*}
 According to (\ref{eq:34}), $p_0$ can be positive or negative.   Hence  for given $p_\mu$, $\mu=1,2,3$,  there are four linearly independent solutions of the form (\ref{eq:31}) for positive $p_0$ and four for negative
 $p_0$.

 It should be noted that in the continuum case there are also eight linearly independent plane-wave solutions of the Joyce equation for a given momentum vector \cite{J}. Moreover, the conditions (\ref{eq:38}) and (\ref{eq:39}) are the same  in both the continuum and discrete cases \cite{B1}.
\begin{proposition}
Let the form (\ref{eq:31})  be a solution of the discrete Joyce equation. If
\begin{equation}\label{eq:40}
\alpha^{0}=-i\alpha^{12}, \qquad \alpha^{13}=i\alpha^{23},
\end{equation}
then  $A\psi$  satisfies the discrete Hestenes equation.
\end{proposition}
\begin{proof}
By Proposition~4, we have
\begin{equation*}
A\psi=A\psi P_{+12}+A\psi P_{-12},
\end{equation*}
where
$A\psi P_{+12}$  satisfies the discrete Hestenes equation.

Let us compute $A_{+} P_{-12}$. We have
 \begin{eqnarray*}
A_{+} P_{-12}=\frac{1}{2}(A_{+}-iA_{+}e_{12})=\frac{1}{2}(A_{+}-i\alpha^{0}e_{12}+i\alpha^{12}x+i\alpha^{13}e_{23}-i\alpha^{23}e_{13})\\
=\frac{1}{2}\big((\alpha^{0}+i\alpha^{12})x+(\alpha^{12}-i\alpha^{0})e_{12}+(\alpha^{13}-i\alpha^{23})e_{13}+(\alpha^{23}+i\alpha^{13})e_{23}\big).
\end{eqnarray*}
Applying (\ref{eq:40}) gives $A_{+} P_{-12}=0$. From (\ref{eq:38}) it follows that $A_{-} P_{-12}=0$ also. This gives $AP_{-12}=0$. Since the 0-form  $\psi$  (\ref{eq:29}) commutes with any form, we thus get
\begin{equation*}
A\psi=A\psi P_{+12}.
\end{equation*}
\end{proof}
\begin{corollary}
The conditions (\ref{eq:40}) are equivalent to the following conditions
\begin{equation*}
\alpha^{01}=i\alpha^{02}, \qquad \alpha^{03}=-i\alpha^{4}.
\end{equation*}
\end{corollary}
\begin{corollary}
 If
\begin{equation*}
\alpha^{0}=i\alpha^{12}, \qquad \alpha^{13}=-i\alpha^{23},
\end{equation*}
or equivalently,
\begin{equation*}
\alpha^{01}=-i\alpha^{02}, \qquad \alpha^{03}=i\alpha^{4},
\end{equation*}
then  $A\psi$  satisfies the discrete Hestenes equation with a reversed mass
sign.
\end{corollary}
\begin{proof}
Under these conditions, the part $A\psi P_{+12}$  vanishes and  $A\psi=A\psi P_{-12}$ satisfies the discrete Hestenes equation with a reversed mass
sign.
\end{proof}

As a final remark,  it is worth pointing out that in the case of a discrete version of the real plane wave solution to the Hestenes equation the components of (\ref{eq:29}) should be given in the  form
\begin{equation*}
  \psi_k=(x-p_0e_{12})^{k_0}(x-p_1e_{12})^{k_1}(x-p_2e_{12})^{k_2}(x-p_3e_{12})^{k_3},
 \end{equation*}
 where $x$ and $e_{12}$ are given by (\ref{eq:21}) and $p_\mu\in\mathbb{R}$. This is the subject of current work in progress.
 
\end{document}